\documentclass{article}
\usepackage{ijcai21}

\usepackage{bm,bbm}
\usepackage{multirow}
\usepackage{ifthen}
\usepackage{float}
\usepackage{balance}
\usepackage{xspace}
\usepackage{times}
\usepackage{soul}
\usepackage{url}
\usepackage{natbib}
\usepackage[hidelinks]{hyperref}
\usepackage[utf8]{inputenc}
\usepackage[small]{caption}
\usepackage{graphicx}
\usepackage{amsmath}
\usepackage{amsthm}
\usepackage{booktabs}
\usepackage{algorithm}
\usepackage{algorithmic}
\providecommand{\IS}[1]{\ensuremath{x_{#1}}\xspace} % investment strategy
\providecommand{\EIS}[1]{\ensuremath{x^*_{#1}}\xspace} % equilibrium investment strategy 
\providecommand{\ISV}{\ensuremath{\bm{x}}\xspace}  % investment strategy vector
\providecommand{\EISV}{\ensuremath{\bm{x^*}}\xspace} % vector of equilibrium investment strategies
\providecommand{\Neigh}[2][]{\ensuremath{%
    \ifthenelse{\equal{#1}{}}{\mathcal{N}_{#2}} {\mathcal{N}_{#2}^{(#1)}}}\xspace} % neighbors
\providecommand{\NeInv}[2][]{\ensuremath{%
    \ifthenelse{\equal{#1}{}}{n_{#2}}{n_{#2}^{(#1)}}}\xspace} % number of investing neighbors. Optional parameter is the strategy profile.
\providecommand{\ENeInv}[2][]{\ensuremath{%
    \ifthenelse{\equal{#1}{}}{n^\ast_{#2}}{n^\ast_{#2}^{(#1)}}}\xspace} % number of equilibrium investing neighbors. Optional parameter is the strategy profile.
\providecommand{\InvCost}[1]{\ensuremath{c_{#1}}\xspace} % cost for investing
\providecommand{\EdgeCost}[1]{\ensuremath{\gamma_{#1}}\xspace} % cost for adding/removing an edge
\providecommand{\EdgeSetCostTmp}[1]{\ensuremath{\hat{\gamma}^{(#1)}}\xspace} % cost for adding/removing an edge
\providecommand{\InvDeg}[1]{\ensuremath{D_#1}\xspace} % investment degree set for a player

\providecommand{\Threshold}[1]{\ensuremath{\theta_{#1}}\xspace} % the coefficients vector c (c^T * x)
\providecommand{\InvPlayer}{\ensuremath{\mathcal{I}}\xspace} % the set of investment players under the desired PSNE
\providecommand{\Pack}[1]{\ensuremath{\mathcal{P}_#1}\xspace}

\providecommand{\Degree}[2][]{\ensuremath{
    \ifthenelse{\equal{#1}{}}{d_{#2}}{d_{#1}(#2)}
}\xspace} % degree of a node in a graph e.g., \Degree[G]{i} = d^G_i   and   \Degree{i} = d_i
\providecommand{\NDDS}[1][]{\ensuremath{
    \textsc{NDDS}
}\xspace} 

\providecommand{\DiEdge}[2]{\ensuremath{ #1 \rightarrow #2 }\xspace} % \DiEdge{i}{j}: i -> j
\providecommand{\NUM}[1][]{\ensuremath{%
\ifthenelse{\equal{#1}{}}{K}{K^{#1}}}\xspace}
\providecommand{\EdgeSet}[2][]{\ensuremath{%
\ifthenelse{\equal{#1}{}}{\mathcal{S}^{(#2)}}{\mathcal{S}^{#1,#2}} }\xspace} % each set represents the modification on a subset of edges 
\providecommand{\EdgeSetSign}[2][]{\ensuremath{%
\ifthenelse{\equal{#1}{}}{\sigma^{(#2)}}{\sigma^{#1,#2}} }\xspace} % direction in which the edges are changed
\providecommand{\EdgeSetMatrix}[2][]{\ensuremath{%
\ifthenelse{\equal{#1}{}}{M^{(#2)}}{M^{#1,#2}}}\xspace} % each set represents the modification on a subset of edges 
\providecommand{\EdgeSetEntry}[4][]{\ensuremath{%
\ifthenelse{\equal{#1}{}}{m^{(#2)}_{#3,#4}}{m^{#1,#2}_{#3,#4}} }\xspace} % each set represents the modification on a subset of edges 
\providecommand{\EdgeSetCost}[2][]{\ensuremath{%
\ifthenelse{\equal{#1}{}}{\gamma^{(#2)}}{\gamma^{#1,#2}}}\xspace} % the cost of modifying the edges in an edge set
\newcommand{\GFUNC}[1][]{\ensuremath{%
    \ifthenelse{\equal{#1}{}}{g_i}{g_{#1}}}\xspace}
\newcommand{\GFunc}[2][]{\ensuremath{\GFUNC[#1](#2)}\xspace}

\providecommand{\AltUt}[3][]{\ensuremath{%
\ifthenelse{\equal{#1}{}}{U_{#2}(#3)}{U^{(#1)}_{#2}(#3)}}\xspace}
\providecommand{\ALTM}{\ensuremath{A}\xspace} % altruism matrix/network
\providecommand{\ALTMIN}{\ensuremath{A^{\text{in}}}\xspace} % altruism matrix/network
\providecommand{\ALT}{\ensuremath{a}\xspace}
\providecommand{\Alt}[3][]{\ensuremath{%
\ifthenelse{\equal{#1}{}}{\ALT_{#2,#3}}{\ALT^{(#1)}_{#2,#3}}}\xspace}
\providecommand{\AltIn}[2]{\ensuremath{\ALT^{\text{in}}_{#1,#2}}\xspace}

\providecommand{\BNPG}[1]{\ensuremath{%
\ifthenelse{\equal{#1}{}}{\textsc{BNPG}}{\textsc{BNPG}(#1)}
}\xspace}

\providecommand{\ALLPSNE}[1]{\ensuremath{\mathcal{E}(#1)}\xspace}
\providecommand{\MODV}{\ensuremath{\bm{v}}\xspace}
\providecommand{\MOD}[1]{\ensuremath{v_{#1}}\xspace}
\providecommand{\ANM}{ANM\xspace}

\providecommand{\GraphIn}{\ensuremath{G^{\text{in}}}\xspace}
\providecommand{\GraphOut}{\ensuremath{G}\xspace}
\providecommand{\EdgeIn}{\ensuremath{E^{\text{in}}}\xspace}
\providecommand{\EdgeOut}{\ensuremath{E}\xspace}
\providecommand{\NDSSGraphIn}{\ensuremath{\hat{G}^{\text{in}}}\xspace}
\providecommand{\NDSSEdgeIn}{\ensuremath{\hat{E}^{\text{in}}}\xspace}
\providecommand{\NDSSGraphOut}{\ensuremath{\hat{G}}\xspace}
\providecommand{\NDSSEdgeOut}{\ensuremath{\hat{E}}\xspace}
\providecommand{\EdgeChange}{\ensuremath{\Gamma}\xspace}
\providecommand{\NDSSEdgeChange}{\ensuremath{\hat{\Gamma}}\xspace}

\providecommand{\Diff}[2]{\ensuremath{ \Delta_{#1}^{#2} }\xspace} % \Diff{i}{+}
\providecommand{\MarginB}{\ensuremath{ \Delta  }\xspace} % uniform margin benefits
\providecommand{\SetCard}[1]{\ensuremath{| #1 |}\xspace}
\providecommand{\SET}[1]{\ensuremath{\{ #1 \}}\xspace}

\providecommand{\Set}[2]{\ensuremath{\SET{#1 \mid #2}}\xspace}

\providecommand{\Ceiling}[1]{\ensuremath{\lceil {#1} \rceil}\xspace}
\providecommand{\Floor}[1]{\ensuremath{\lfloor {#1} \rfloor}\xspace}

\DeclareSymbolFont{AMSb}{U}{msb}{m}{n}
\DeclareMathSymbol{\N}{\mathord}{AMSb}{"4E}
\DeclareMathSymbol{\B}{\mathord}{AMSb}{"42}
\DeclareMathSymbol{\Z}{\mathord}{AMSb}{"5A}
\DeclareMathSymbol{\R}{\mathord}{AMSb}{"52}

\providecommand{\Omit}[1]{{}}

\newtheorem{theorem}{Theorem}[section]
\newtheorem{lemma}[theorem]{Lemma}

\newtheorem{proposition}[theorem]{Proposition}

\newtheorem{definition}{Definition}[section]

\theoremstyle{remark}

\newcommand{\LPlabel}{}

\newenvironment{LP}[3][]{%
\renewcommand{\LPlabel}{#1}
\ifthenelse{\equal{\LPlabel}{}}{%
\[ \begin{array}{l}
\mbox{#2} \;\;#3 \;\; \mbox{subject to} \\
          \begin{array}[t]{ll}%
}{%
\begin{equation} \begin{array}{l}
\mbox{#2} \;\;#3 \;\; \mbox{subject to} \\
          \begin{array}[t]{ll}%
}}{%
\ifthenelse{\equal{\LPlabel}{}}{%
\end{array} \end{array} \]}{%
\end{array} \end{array} \label{\LPlabel} \end{equation}}%
}

\usepackage[suppress]{color-edits}
\addauthor[David]{dk}{red}
\addauthor[Eugene]{yv}{green}
\addauthor[Sixie]{sy}{blue}

\usepackage{enumitem}
\setitemize{noitemsep,topsep=0pt,parsep=0pt,partopsep=0pt}

\title{IJCAI--21 Formatting Instructions}

\author{
  Sixie Yu$^1$
  \and
  David Kempe$^2$
  \and
  Yevgeniy Vorobeychik$^{1}$
  
  \affiliations
  $^1$Washington University in St. Louis\\
  $^2$University of Southern California\\

  \emails
  $^1$\{sixie.yu,yvorobeychik\}@wustl.edu,
  $^2$david.m.kempe@gmail.com
}

\begin{document}
\title{Altruism Design in Networked Public Goods Games}
% \author{}
\maketitle

\begin{abstract}
Many collective decision-making settings feature a strategic tension
between agents acting out of individual self-interest and promoting a common good.
These include wearing face masks during a pandemic, voting, and vaccination.
Networked public goods games 
capture this tension, with networks encoding strategic interdependence among agents.
Conventional models of public goods games posit solely individual self-interest as a motivation, even though altruistic
motivations have long been known to play a significant role in agents' decisions.
We introduce a novel extension of public goods games to account for
altruistic motivations by adding a term in the utility function that
incorporates the perceived benefits an agent obtains from the welfare
of others, mediated by an altruism graph.
Most importantly, we view altruism not as immutable, but rather as a lever for promoting the common good.
Our central algorithmic question then revolves around the
computational complexity of modifying the altruism network to achieve desired public goods game investment profiles.
We first show that the problem can be solved using linear programming
when a principal can fractionally modify the altruism network.
While the problem becomes in general intractable if the principal's
actions are all-or-nothing, we exhibit several tractable special cases.
\end{abstract}

\section{Introduction}\label{sec:intro}
Individuals in a collective decision-making environment often experience the following type of scenario.
Each individual can decide whether or how much effort to invest for the common good; %everyone else in the environment
many others may benefit from the efforts, 
but the cost of the investment is incurred by the individual.
Examples of such scenarios include decisions whether or not to wear a mask in a pandemic, vaccinate, or invest in security.
%, protect a natural environment, or vote for a plan that has long-term benefits and short-term costs.
The outcomes of such scenarios are often highly suboptimal from a societal point of view: mask-wearing suggestions are flaunted, societies remain undervaccinated, and security measures are not taken.
%, beaches and forests are littered, and long-term planning is sacrificed in favor of short-term concerns. 
At the heart of this breakdown is that while individuals are ``connected'' in the sense that their actions affect one another, they are often disconnected ``socially,'' in the sense that they do not experience the utility gain/loss of those affected by their actions.
In economic terms, actions have \emph{externalities} on other players, which are, by definition, not \emph{internalized}. 

Indeed, the outcomes of such scenarios tend to be significantly different when the individuals form a more tightly knit community. Within families, groups of friends, or small villages, individuals frequently take actions, at a cost to themselves, which primarily benefit others. Similarly, societies with a stronger sense of ``duty'' towards fellow citizens tend to witness more compliance in all of the above-mentioned examples.
Not surprisingly then, campaigns to encourage individual effort (e.g., ``Wear a mask --- save a life!'') tend to appeal to notions of altruism and duty, attempting to get individuals to internalize some of their externalities, if only psychologically.
%\syedit{
%	In fact, influencing attitudes is the only practical option in many settings, for example, in seasonal flu vaccination, it’s typically impractical to change interaction patterns (who can be infected by whom), and we would prefer for people to choose to vaccinate from a combination of selfish and altruistic motivations.
%}

If the goal of campaigns is to encourage altruistic behavior, an important question is what type of campaign is most effective.
Should a principal, aiming to achieve a societally desirable outcome, try to appeal to a generic sense of ``duty towards your fellow citizens,'' try to strengthen the social ties within a small neighborhood, or focus on building a few strong ties between some key individuals?
Can the question of how best to \emph{build} or \emph{change} altruism in a society be approached algorithmically, and are the resulting questions tractable or intractable?
This is the high-level question we investigate in the present paper.

%\subsection{Modeling Interactions and Altruism}

The question of how to build altruism networks is meaningful in a variety of strategic settings.
%and interesting for many different games.
We focus on
%\emph{public goods games}~\citep{mas-collel:whinston:green}, or --- more specifically --- on
\emph{networked public goods games}~\citep{bramoulle:kranton:public-goods,bramoulle:kranton:damours:strategic,PublicGoods,Yu20}, motivated by the real-world scenarios discussed earlier (e.g., encouraging mask wearing).
In networked public goods games, the benefits of an individual's effort are reaped by those with whom the individual interacts, encoded by a network on the individuals.\footnote{Public goods games can be viewed as the special case in which the network is complete.}
%More
Specifically,
%the prior literature as well as the present work study a form in which
an individual's utility depends on 1) her own investment decision, and 2) the \emph{aggregate} investment from her direct neighbors in the network.
%(Formal definitions are given in Section~\ref{sec:prelim}.)

In most conventional models of games, including public goods games, it is assumed that agents are driven solely by their individual interests.
This assumption is nearly always violated in behavioral studies of public goods games~\citep{ledyard:public-goods,levine:altruism-spitefulness}.
%, where agents have been observed to exhibit altruistic behavior and invest for promoting the common good, even if the decisions are sub-optimal in terms of their self-interest.
While there are many different ways to model altruistic behavior, one natural way was proposed by \citet{ledyard:public-goods}: the utility of a player $i$ is a linear combination of an \emph{egocentric utility} term, which is the direct benefit to $i$, and an \emph{altruistic} term, which is a sum of egocentric utilities of other players $j$, weighted by the strength \Alt{i}{j} of altruism that $i$ feels for $j$.
In most prior work of this kind, \Alt{i}{j} was modeled as a constant for all players $i$ and $j$.
A more general variation by \citet{meier:oswald:schmid:wattenhofer} considered an altruism network in vaccination games, but assumed that the altruism graph is identical with the graph representing strategic dependence, as well as that altruism weights are identical for all edges.
%setting with non-uniform \Alt{i}{j} was considered in the context of vaccination games by \citet{meier:oswald:schmid:wattenhofer}, who assumed that the graph of who cared about whom was the same as the graph of who could infect whom, and that all existing edges had an altruism weight of \ALT.
Naturally, many settings call for more fine-grained models in which the weights can be different: for example, parents typically care more about their children's welfare than that of strangers.

While the focus of past work on altruism in games has been on its equilibrium effects, 
our point of departure is to consider the altruism network itself as (partially) under the control of a principal.
%As discussed earlier, public outreach campaigns, community meetings, personal introductions, or (on the flipside) disinformation campaigns can be used to strengthen or weaken an altruism network at different scales, for typically different costs.
In other words, we view altruistic motivations as a lever that can be adjusted to promote the common good, for example, through public outreach campaigns, community meetings, and personal introductions.
%In particular, when the investment profile of a public goods game at equilibrium is not desirable (e.g., inefficient or unfair), a policymaker may be interested in modifying the altruism network so as to induce a more desirable investment profile.\footnote{%
%Another more commonly used lever is a change in the reward structures and rules of encounters of the game. This line of work is the heart of work in mechanism and market design~\citep{nisan:roughgarden:tardos:vazirani}, and is orthogonal to our work.}
%In this paper,
Specifically, we propose a model of modifying altruism networks, with the goal of inducing a target investment profile by the agents.
We consider three variants of the altruism network: weighted, directed, and undirected.
We show that even for very complex available actions, the problem can be solved efficiently using linear programming when the principal has fine-grained control over the extent to which actions are taken.
%not over \emph{whether} actions are taken, but \emph{to what extent}.
When the principal can only control \emph{which} actions are taken,
%but not to which extent (i.e., each action poses an all-or-nothing choice),
the problem becomes NP-complete, even when each action affects only a single edge in the altruism network.
However, when the altruism network is directed, we show that the problem is tractable in a broad array of special cases by reductions to the (tractable cases of) the \textsc{Knapsack} problem.
We also leverage this connection to exhibit an FPTAS for the general case.
When the altruism network is undirected, the hardness results apply even for much more restrictive special cases.
However, we show that the problem is tractable when the benefits from investment are linear and uniform, by a non-trivial reduction to the problem of \textsc{network design for degree sets (NDDS)} introduced by \citep{kempe2020inducing}, who also showed that it can be solved in polynomial time.
%minimum-cost matching.

Our problem of designing an altruism graph to achieve target equilibrium outcomes is, indeed, conceptually related to \citet{kempe2020inducing}, who study the problem of designing the \emph{strategic} network in networked public goods games.
The main rationale for shifting focus to designing \emph{altruism graphs} is that strategic networks are often difficult to change.
For example, in a pandemic, it is difficult to directly affect contacts among individuals, as these are ultimately the products of individual choices (e.g., even lockdowns may be ineffective if individuals are non-compliant, except through levels of enforcement that are often viewed as unacceptable by the population).
In contrast, it can be significantly easier to try to impact decisions indirectly by evoking altruistic motivations in people.
From a technical perspective, the problem of altruism design impacts utilities linearly, in contrast to the design of strategic networks; however, it is also distinct from the linear special case in \citet{kempe2020inducing}, where the marginal impact of each neighbor on a player's utility is identical, in contrast to altruism design, where these differ.
  
\paragraph{Related Work }
Our work is related to four lines of research: graphical games, altruism modeling, mechanism and market design, and network design.
Graphical games encode sparsity in the interdependence of player utility functions using a graph~\citep{kearns2013graphical,shoham2008multiagent}, with networked public goods games an important class of such models~\citep{bramoulle:kranton:public-goods,galeotti2010network,grossklags2008security,Yu20}.
%structure in the players' utility functions which restrict the scope of utility dependence on other players' actions~\citep{shoham2008multiagent}.
%An important class of these are \emph{graphical games}, where a player's utility only depends on the actions of her network neighbors~\citep{kearns2013graphical}.
%Networked public goods games are one important example of graphical games, with utilities only depending on the actions by a player's network neighbors~\citep{bramoulle:kranton:public-goods,galeotti2010network,grossklags2008security,Yu20}.
%An important assumption in networked public goods games is that players act purely out of self-interest.
%Our work, however, introduces a novel extension of networked public goods game to account for altruistic motivations.
A conventional assumption in such games is that agents act to exclusively promote their own interest.
However, considerable experimental evidence exists that even games with this payoff structure elicit altruistic motivations among human subjects~\citep{dong2016dynamics,levine:altruism-spitefulness}.
This, in turn, led to a series of approaches to model altruism in a variety of games, including public goods games, which are of direct interest here~\citep{ledyard:public-goods,dong2016dynamics}, inoculation games~\citep{meier:oswald:schmid:wattenhofer}, routing games~\citep{AltruisticRouting}, and congestion games~\citep{AtomicGames,AltruismPoA}.

%Another closely related thread is altruism modeling.
%altruistic motivations have been observed both in practice and in theoretical experiments.
%\citeauthor{ledyard:public-goods}~\shortcites{ledyard:public-goods} and \citeauthor{levine:altruism-spitefulness}~\shortcites{levine:altruism-spitefulness} studied altruistic motivations through behavioral studies of public goods game. 
%\citeauthor{meier2008windfall}~\shortcites{meier2008windfall} modeled altruism through a social network where a player's utility partially depends on her neighbors' utilities.
A typical way that altruism is captured in prior literature is by either adding a social welfare term to utility functions~\citep{ledyard:public-goods}, or introducing a parameter that governs the extent to which agents care about their social network neighbors~\citep{meier:oswald:schmid:wattenhofer}.
Our model is distinct in that it allows altruism to be relationship-dependent, a property we model by an altruism network.
Moreover, our goal is to \emph{modify} an altruism network to achieve a target equilibrium (e.g., one that maximizes social welfare).
%is that altruistic motivations are captured as a social welfare term added to a utility function. 
%The novelty of our work are as follows:   
%1) we propose a graph-based model of altruism, with the capability of capturing finer granularity exhibited in real altruistic behavior; 
%2) we propose a model of designing altruistic motivations so as to induce (socially) preferrable equilibrium outcomes;
%3) we conduct an algorithmic study centering around the computational complexity of designing altruistic motivations. 

\syreplace{
	Mechanism and market design (e.g., \citep{nisan:roughgarden:tardos:vazirani}) also aim to change the parameters of a game to induce  desirable equilibrium outcomes.
	There are various ways in which the game's parameters are changed; key approaches include the design of market structure (e.g., matching market mechanism~\citep{Haeringer18}), payment rules as in traditional mechanism design~\citep{nisan:roughgarden:tardos:vazirani}, or the structure of information available to the players~\citep{dughmi:information-structure}.
	We introduce altruism network design as a novel lever for aligning incentives with public good.
}{
	Mechanism and market design also aim to change the parameters of a game to induce  desirable equilibrium outcomes~\citep{nisan:roughgarden:tardos:vazirani,Haeringer18,dughmi:information-structure}.
	We introduce altruism network design as a novel lever for aligning incentives with public good.
}

\syreplace{
	The last relevant line of research is network design.
 	A particularly related thread in network design is to study the effects of network modification on equilibrium outcomes or welfare.
	\citet{kempe2020inducing} initiated an algorithmic study of target network modification with the goal of inducing  equilibria of a particular form, but focused on the network representing strategic dependencies, rather than altruism.
	In particular, as discussed above, \citet{kempe2020inducing} design the network that impacts strategic interdependence; in contrast, we design the network that impacts altruism portion of player utilities.
	The crucial difference is in the impact of network design on the utilities: in \citep{kempe2020inducing}, it is mediated by the non-linear benefits term that only depends on the number of investors, whereas in our work, the impact is linear, but depends on the identities of the altruism graph neighbors.

	%Since the impact that the network has on player utilities is completely different between the two models, the study of \citet{kempe2020inducing} does not provide much insight for our problem. 
	Other studies in this line of research include the following: \citet{bramoulle:kranton:public-goods} considered the addition of a single edge to the underlying network of a public goods game and studied  how the addition affects welfare.
	\citet{galeotti2010network} also studied the effects of modifying the underlying network of a public goods game on equilibrium behavior and welfare. 
}{
	The last relevant line of research is network design.
	A particularly related thread in network design is to study the effects of network modification on equilibrium outcomes or welfare~\citep{kempe2020inducing,bramoulle:kranton:public-goods,galeotti2010network}.
	Another related thread is to alter a network in order to effect a variety of different outcomes for different types of games~\citep{Sheldon10,CTPEFF:edge-manipulation,ghosh:boyd:growing,TPEFF:gelling,bredereck2017manipulating,sina2015adapting,matteo2020manipulating,amelkin2019fighting,garimella:morales:gionis:mathioudakis}.
}
% Another related thread is to alter a network in order to effect a variety of different outcomes
% including, e.g., maximizing the spread of cascades~\citep{Sheldon10}, making a network more or less connected~\citep{CTPEFF:edge-manipulation,ghosh:boyd:growing,TPEFF:gelling}, manipulating the majority opinion in the converging state of simple information diffusion dynamics~\citep{bredereck2017manipulating}, 
% manipulating the outcome of an election through altering network structures~\citep{sina2015adapting,matteo2020manipulating},
% reducing opinion control~\citep{amelkin2019fighting} or decreasing opinion polarization~\citep{garimella:morales:gionis:mathioudakis}.
% These share the high-level goal of inducing certain outcomes, but the specific optimization goals and algorithmic approaches are vastly different.

\section{Networked Public Goods Games and Altruism}\label{sec:prelim}
We study altruism in \emph{binary networked public goods games (BNPGs)}, which are an important variant of \emph{public goods games} studied extensively in prior literature~\citep{bramoulle:kranton:public-goods,galeotti2010network,grossklags2008security,suri2011cooperation,dong2016dynamics,kempe2020inducing,Yu20}.
%\dkcomment{We need to add more citations to papers by others to corroborate the claim that these are important, and also to not de-anonymize the paper.}
%\sycomment{Check out the newly added citations. In particular, the paper by \citep{dong2016dynamics} also studied how other-regarding preferences in public goods game affect players' contributions. }
We begin by formally describing BNPGs and pure strategy Nash equilibria, the solution concept we focus on.
We then discuss a natural model of altruism in games, and its application to the specific case of BNPGs.

\noindent{\bf Binary Networked Public Goods Games: }
A \emph{binary networked public goods} (BNPG) game is characterized by the following:
    \begin{enumerate}
        \item A simple, undirected, and loop-free graph $H=(V, E_H)$ in which the nodes $V = \SET{1, 2, \ldots, n}$ are the agents/players, and the edges $E_H$ represent the interdependencies among the players' payoffs. 
        \item A binary strategy space $\SET{0,1}$ for each player $i$.
\footnote{\dkedit{The public goods game literature extensively considers both binary and continuous decisions, and both are natural candidates for considering network design. Here, we focus on binary strategies due to their applicability to decisions such as vaccinations or mask wearing.}}
%          We motivate the importance of considering binary strategies (e.g., mask wearing, vaccination decisions); of course, alternatives are also important, and are natural subjects of future work.
        We interpret the choice of strategy 1 as \emph{investing} in a public good, while choosing 0 is interpreted as non-investment.
        The action of player $i$ is denoted by \IS{i}, and the joint pure strategy profile of all players by $\ISV = (\IS{1}, \IS{2}, \ldots, \IS{m})$. We use $\ISV_{-i}$ to denote a strategy profile that omits player $i$'s strategy.
        \item A non-decreasing utility function $U_i(\IS{i},\ISV_{\Neigh[H]{i}})$ for each player $i$, where $\Neigh[H]{i} = \Set{j}{ (i, j) \in E_H }$ is the set of $i$'s neighbors in the graph $H$.
    \end{enumerate}

    As is common in the literature on public goods games~\citep{bramoulle:kranton:public-goods}, we assume that each player's (egocentric) utility function $U_i$ only depends on the total investment by $i$'s network neighbors.
    To formalize this, we define $\NeInv[H,\ISV]{i} = \sum_{j \in \Neigh[H]{i}} \IS{j}$ as the number of $i$'s neighbors who invest under \ISV.
    We omit $H$, $\ISV$, or both from this notation when they are clear from the context.
Each player $i$'s utility function then has the following form:
  % \vspace{-0.1in}
  \begin{small}
    \begin{align}
    \label{eq:utility-original}
      U_i (\ISV) & = U_i(\IS{i}, \NeInv[\ISV]{i})
      \; = \; \GFunc[i]{\IS{i}, \NeInv[\ISV]{i}} - \InvCost{i} \IS{i}.
    \end{align}
  \end{small}
  % \vspace{-0.1in}

The second term ($-\InvCost{i} \IS{i}$) captures the cost incurred by player $i$ from investing.
As is standard in the public goods games literature, each \GFUNC[i] is assumed to be a non-negative and non-decreasing function of both of its arguments, capturing the positive externality that $i$ experiences from her neighbors' (and her own) investments.
Observe that each function \GFUNC[i] can be represented using $O(n)$ values, so the entire BNPG game (including the graph structure) can be represented using $O(n^2)$ values.

We will consider \emph{pure-strategy Nash equilibria (PSNE)} of BNPGs.
A pure strategy profile $\ISV^*$ is a PSNE if for all $i$, $U_i(x_i, \NeInv[H,\ISV]{i} ) \ge U_i(1-x_i, \NeInv[H,\ISV]{i} )$.
We write \ALLPSNE{\mathcal{G}} for the set of all PSNEs of the game $\mathcal{G}$.

\noindent{\bf Altruistic Motivations in BNPGs: }
%Most models of games, including public goods games, assume that players are entirely self-interested, in the sense that they do not consider the consequences of their decisions on others in evaluating their utility.
%However, this assumption is nearly always violated in behavioral studies of public goods games (e.g., \citep{ledyard:public-goods,levine:altruism-spitefulness}): participants frequently take actions that benefit others when it does not come at too high cost to themselves.
A natural way to model other-regarding utilities is to define a player's utility as a linear combination of her
egocentric utility, defined by Equation~\eqref{eq:utility-original}, and the egocentric utilities of other players.

To formalize this, we can think of the matrix $\ALTM = (\Alt{i}{j})_{i,j}$ as encoding an \emph{altruism network}.
%Crucially, the non-zero entries of \ALTM may not be at the same positions as the edges of $H$.
This captures the central motivation of our work, discussed in the introduction: that the agents who are \emph{affected} by the actions of $i$ may not be the same as the agents that $i$ \emph{cares about}.%
\footnote{In the context of vaccination games, \citet{meier:oswald:schmid:wattenhofer} study the special case in which the friendship network is the \emph{same} as the network of who may transmit a disease to whom.}
The resulting utility function of a player $i$ in our BNPG game model with altruism is
% Specializing Equation~\eqref{eqn:general-altruistic-utility} to our BNPG game model, we obtain the following form for other-regarding utilities:
  % \vspace{-0.05in}
  \begin{small}
    \begin{align}
    \label{eq:utility}
      \AltUt[\ALTM]{i}{\ISV}
      & = \GFunc[i]{\IS{i}, \NeInv{i}} - \InvCost{i} \IS{i}
      % +  \sum_{j \in V} \Alt{i}{j} \GFunc[j]{\IS{j}, \NeInv{j}}.
            +  \sum_{j \in \Neigh[H]{i}} \Alt{i}{j} \GFunc[j]{\IS{j}, \NeInv{j}}.
    \end{align}
  \end{small}
  \vspace{-0.1in}

We denote the BNPG with altruism network \ALTM by \BNPG{\ALTM}.

We note two points about this model:
First, the altruistic term of player $i$'s utility does not include a term for the \emph{investment cost} of player $j$, only the \emph{utility}.
This is inconsequential, as investment decisions by $j$ are not under $i$'s control, so from $i$'s point of view, these cost terms are constants.
%In other words, all equilibria will be the same.
Second, the other-regarding terms have no component in which $j$'s utility due to $i$'s egocentric payoff is recursively considered. Such utilities may be harder to observe, and from a modeling perspective, they can be transformed
%(using a matrix inversion)
into the case we study here; see~\citep{bergstrom:benevolent}.
%See the work of \citet{bergstrom:benevolent} for a discussion and derivation.

In some of the results in later sections, we will specifically want to stress the \emph{network} aspect of altruism.
In that case, we will assume that we are given a (directed or undirected) \emph{altruism graph} $G$, and that $\Alt{i}{i} = 1$ for all $i$, $\Alt{i}{j} = \ALT$ for all $i,j$ for which $G$ contains the edge $(i,j)$, and $\Alt{i}{j} = 0$ for all $i,j$ for which $G$ contains no edge. In other words, the altruism strength of the edges of $G$ is uniform.
When $G$ is undirected, we will refer to the case as \emph{symmetric altruism}; when $G$ is directed, we call it \emph{asymmetric altruism}.

\section{Modifying Altruism Networks}\label{sec:model}
As discussed in the introduction, a major problem with public goods situations is that equilibria can be far from optimal because individuals may not fully internalize the impact of their actions on others.
Therefore, a principal who seeks to steer the network to a better equilibrium might wish for agents to consider others in their decisions.
We consider situations in which the principal can increase or decrease the salience of others in these settings, for example, through introductions, advertising, or community meetings.\footnote{This is separate from, and in addition to, other channels, such as rewarding or punishing certain actions.}
We now formally model this problem as modifying an altruism network to achieve socially desired outcomes.
Our model is as follows.
The principal aims to induce a particular \emph{target investment profile} \EISV.
This allows us to cleanly capture a broad variety of design goals, such as maximizing welfare or achieving fairness, while focusing on the computational issues at the core of our specific problem of altruism design.
%(e.g., a welfare-optimal, Pareto optimal, fair, or otherwise desirable outcome).
To induce the outcome, the principal wants to (minimally) modify the altruism network \ALTMIN to \ALTM such that \EISV is a PSNE of the modified game \BNPG{\ALTM}.\footnote{There may be other PSNE of the game. We implicitly assume that the principal can \emph{suggest} an equilibrium to the agents, who will follow the suggestion unless it is in their best interest to deviate.}

As implied by the preceding discussions, the principal may have at his disposal a number of different actions, affecting the altruism between different sets of pairs of agents, in positive or negative ways. For example, a general appeal to watch out for one another may lead to a small increase in altruism between many pairs of individuals; a community meeting may lead to a stronger increase among a smaller subset, and a personal introduction may introduce one strong edge.
%On the flipside, a disinformation campaign based on fake profiles may lead to a reduction in altruism between two different groups.
%In full generality,
We model such settings by assuming that there are
\NUM \emph{actions}, with \NUM polynomial in $n$.
%\footnote{We assume that \NUM is polynomially bounded in $n$.}.
Each action $k$ has associated with it a set \EdgeSet{k} of affected altruism edges,
%(pairs of individuals),
a cost $\EdgeSetCost{k} \geq 0$,
%Furthermore, it has also associated with it
and a \emph{sign} $\EdgeSetSign{k} \in \SET{-1,1}$, which captures whether the action strengthens or weakens the edges in \EdgeSet{k}.%
\footnote{Typically, a principal would be more likely to want to \emph{strengthen} the altruism network. However, one can easily imagine situations and construct instances in which a weakening of the network is necessary. We therefore aim for more generality in our model.}
The edge set is encoded in the corresponding adjacency matrix $\EdgeSetMatrix{k}$, with entries $\EdgeSetEntry{k}{i}{j}$ which are 1 for $(i,j) \in \EdgeSet{k}$ and 0 otherwise.
The costs \EdgeSetCost{k} measure the monetary expense or effort/time needed to implement the corresponding activity, per unit of change in the altruism.
The principal aims to solve the following problem:

\smallskip
\begin{definition}[Altruism Network Modifications (\ANM)]
\label{def:main-problem-functions}
\textbf{Given:} an altruism network
\ALTMIN, target investment profile \EISV,
and actions $\SET{\EdgeSet{1}, \ldots, \EdgeSet{\NUM}}$ with signs \EdgeSetSign{k} and costs \EdgeSetCost{k}.
\textbf{Goal:} choose a non-negative vector $\MODV \in \R_{\geq 0}^{\NUM}$ of minimum total cost
$\sum_{k=1}^{\NUM} \MOD{k} \EdgeSetCost{k}$
such that the modified game with new altruism network
  \vspace{-0.1in}
  \begin{small}
    \begin{align}
      \ALTM & = \ALTMIN + \sum_{k=1}^{\NUM} \MOD{k} \cdot \EdgeSetSign{k} \cdot \EdgeSetMatrix{k}
    \label{eqn:modified-altruism}
    \end{align}
  \end{small}          
has \EISV as a PSNE, i.e., $\EISV \in \ALLPSNE{\BNPG{\ALTM}}$.
\end{definition}
Here, \MODV is a vector capturing how much the principal spends on each of the available actions. We assume that the different actions are cumulative in their effects on any of the network's edges, and (partially) cancel out when they have opposite signs.
Note that we could additionally truncate the entries of \ALTM so that $0 \le \ALTM \le 1$; this is not consequential for our results, and we proceed with the slightly cleaner model above.
The principal's spending on actions results in a modified altruism network, and his goal is to ensure that the target action profile \EISV becomes one of the equilibria of the modified game.
The principal wants to achieve this goal at minimum cost (which is infinite if the problem is infeasible).%
Since strategies are binary, a target profile \EISV can be equivalently represented by the set of agents who invest under this profile, $\InvPlayer=\InvPlayer_{\EISV}=\Set{i \in V}{\EIS{i} = 1}$.
Whether or not players invest can be completely characterized using a collection of inequalities.
In these inequalities, what ultimately determines the decision is the agent's marginal value from investing.
This marginal value has two elements: first, the agent's own marginal benefit, $\Delta_{x_i} \GFunc[i]{n_i} := \GFunc[i]{1, n_i} - \GFunc[i]{0, n_i}$, and second, the marginal benefit that $i$ can obtain from altruism towards $j$, which is $\Diff{j}{-} := \GFunc[j]{x_j, n_j} - \GFunc[j]{x_j, n_j-1}$ for agents $i$ who invest, and $\Diff{j}{+} :=\GFunc[j]{x_j, n_j+1} - \GFunc[j]{x_j, n_j}$ for those agents $i$ who do not.
We further define $\Threshold{i} := c_i - \Delta_{x_i}\GFunc[i]{n_i}$ for each agent $i$.
We then obtain the following equivalent characterization of a PSNE of \BNPG{\ALTM}, in which a set \InvPlayer of players invest, and the rest do not:
\vspace{-0.1in}
  \begin{small}
    \begin{equation} 
    \label{E:psne}
      \begin{aligned}
        \sum_{j \in \Neigh[H]{i}} \Alt{i}{j} \Diff{j}{-}
          & \geq \Threshold{i} \quad \text{ if } i \in \InvPlayer \\
        \sum_{j \in \Neigh[H]{i}} \Alt{i}{j}  \Diff{j}{+}
          & \leq \Threshold{i} \quad \text{ if } i \notin \InvPlayer,
      \end{aligned}
    \end{equation}
  \end{small}
% \vspace{-0.1in}

with \Alt{i}{j} the strength of $i$'s altruism for $j$ under \ALTM.

The marginal benefit functions $\Delta \GFUNC:=(\Diff{i}{-}, \Diff{i}{+})$ are important parameters in our model; they can be restricted by putting limitations on \GFUNC, e.g., $\Delta \GFUNC$ is bounded by a polynomial in $n$ iff \GFUNC itself is.
We consider three possible restrictions:

% instances of NDA
%according to various properties of $\Delta \GFUNC$ and \EdgeSetCost{i}.
%As we will see, these properties have different algorithmic implications.
%We start with the properties of \EdgeSetCost{i}.
%In particular, we consider two types of edge costs: general and uniform \EdgeSetCost{i}.
%The two types subsume the edge costs considered in prior studies~\citep{bramoulle:kranton:public-goods,kempe2020inducing}.
%Next, recall that $\Delta \GFUNC[j] :=(\Diff{j}{-}, \Diff{j}{+})$,  we consider the following properties of $\Delta \GFUNC$:

\begin{description}
\item[general $\GFUNC$:]  $\Diff{i}{-}$ and $\Diff{i}{+}$ are arbitrary.
  %non-negative, as long as they are non-negative.
    \item[polynomial $\GFUNC$:] When the \GFUNC are bounded by a polynomial in $n$, $\Diff{i}{-}$ and $\Diff{i}{+}$ are also bounded by a polynomial in $n$. % (the number of players).
%    This intrinsically captures the assumption that marginal utilities are not very large numbers. 
    \item[uniform linear and separable $\GFUNC$:]
      When all \GFUNC[i] are of the form $\GFunc[i]{\IS{i}, n_i} = h_i(\IS{i}) + \MarginB \cdot n_i$ for some possibly idiosyncratic function $h_i$ and some (common) constant $\MarginB$, then
        %there exists a constant \MarginB such that
        $\Diff{i}{-} = \Diff{i}{+} = \MarginB$ for all $i$ and all values of $n_i$.
      % This assumes that players' marginal utilities are homogeneous, i.e., $\Delta \GFUNC[1] = ,\ldots, \Delta \GFUNC[n]$. Note that the $\Diff{j}{-}$ and $\Diff{j}{+}$ in each $\Delta \GFUNC$ do not need to be the same.
\end{description}

%The uniformity can be realized, e.g., when $H$ is a regular graph, the desired PSNE is that all players invest, and the functions $\GFUNC$ admit the same functional form.
%Although it seems restrictive, we argue that the uniformity is pivot from an algorithmic perspective.
%As we will see, it is crucial for obtaining efficient algorithms when $G'$ is undirected.
%In addition, in many cases it is indeed desirable to have all players invest, e.g., wearing masks, investing effort for outreaching, etc.
%Moreover, it is a standard assumption in network games that a certain function \GFUNC is shared by all players~\citep{bramoulle:kranton:public-goods,galeotti2010network,ballester2006s}. 
%We use \NDA{uniform}{general} to represent an instance of NDA with uniform $\Delta \GFUNC$ and general $\EdgeSetCost{i}$.
%Other instances with different $\Delta \GFUNC$ and $\EdgeSetCost{i}$ are similarly represented.

%In the rest of the paper we present an algorithmic study for  NDA.
%The presentation consists of two parts.
%In the first part we focus on weighted $G'$ and describe a Linear Programming (LP) formulation.
%In the second part we switch the focus to  unweighted $G'$, where we show concrete examples of the $k$ sets $\EdgeSet{1}, \ldots, \EdgeSet{k}$ and study the resulting algorithmic questions.
%The study includes both hardness results and tractable cases.
%For tractable cases we provide polynomial-time algorithms.
%For hard cases, in addition to the hardness proofs, we present approximation algorithms with guarantee.

\section{An LP for Fractional Modifications}\label{sec:weighted}

%\providecommand{\AuxiVar}[2]{\ensuremath{ z^#2_#1 }\xspace} % auxiliary variables to get rid of the absolute values in the obj. function of the LP, e.g., \AuxiVar{i}{1} = z^1_i
%\providecommand{\Auxi}{\ensuremath{ \bm{z} }\xspace} % collectively represent all auxiliary variables.

%In this section we present an LP formulation for the problem NDA.
%One significance of the LP is that its modeling is straightforward and 
%it is solvable in polynomial time, which paves its way to real-world applications. 
%In what follows we focus on undirected $G'$.
%It is direct to generalize the formulation to directed $G'$.

We begin by considering the variant of the problem in which the principal can spend fractionally on each action, i.e., $\MODV \in \R_{\geq 0}^{\NUM}$.
In that case, the principal's optimal strategy can be found using a straightforward linear program.
The decision variables are simply the principal's investments $\MOD{k} \geq 0$.
For ease of notation, we also define variables \Alt{i}{j} for the altruism from $i$ towards $j$ resulting from the modifications.
Given the available actions and costs, the relationship between \Alt{i}{j} and \MOD{k} is exactly characterized by Equation~\eqref{eqn:modified-altruism}, which is linear in the \MOD{k}.
Next, notice that given the BNPG and the desired PSNE \EISV, the set \InvPlayer and all relevant constants in Equation~\eqref{E:psne} (that is $\Diff{i}{-}, \Diff{i}{+}, \Threshold{i}$) can be immediately computed.
Therefore, we obtain the following linear program for finding the optimal
spending strategy for the principal:
	\vspace{-0.03in}
	\begin{small}
		\begin{LP}[eq:weighted-LP]{Minimize}{\sum_{k=1}^{\NUM} \MOD{k} \cdot \EdgeSetCost{k}}
		\Alt{i}{j} = \AltIn{i}{j} + \sum_{k=1}^{\NUM} \MOD{k} \cdot \EdgeSetSign{k} \cdot \EdgeSetEntry{k}{i}{j}
		    & \text{ for all } i, j\\
		\sum_{j \in \Neigh[H]{i}} \Alt{i}{j} \Diff{j}{-} \geq \Threshold{i}
		    & \text{ for all } i \in \InvPlayer\\
		\sum_{j \in \Neigh[H]{i}} \Alt{i}{j} \Diff{j}{+} \leq \Threshold{i}
		    & \text{ for all } i \notin \InvPlayer\\
		\MOD{k} \geq 0 & \text{ for all } k.
		\end{LP}
	\end{small}

%Solving this LP lets the principal find the optimum spending
%strategy.

Notice that we could fairly straightforwardly generalize the LP to
even deal with more general actions, namely, allowing an action $k$ to
affect different edges in different ways by allowing
\EdgeSetMatrix{k} to have entries that are not all equal.

\section{Graph-Based Modifications}\label{sec:unweighted}

\Omit{
\begin{table*}[htb]
\small
\centering
\begin{tabular}{@{}ccccc@{}}
\toprule
                                                          &                             & general $\GFUNC$ & polynomial $ \GFUNC$ & uniform $ \GFUNC$ \\ \midrule
\multicolumn{1}{c|}{\multirow{2}{*}{asymmetric altruism}} & general \EdgeSetCost{k}      & hard                    & poly                 & poly                    \\
\multicolumn{1}{c|}{}                                     & polynomial \EdgeSetCost{k}      & poly                    & poly                 & poly
                  \\ \midrule
\multicolumn{1}{c|}{\multirow{2}{*}{symmetric altruism}}  & general \EdgeSetCost{k} & hard                    & hard                 & poly                  \\
\multicolumn{1}{c|}{}                                     & polynomial \EdgeSetCost{k} & hard                    & hard                 & poly
                    \\ \bottomrule
\end{tabular}
\caption{An overview of our results for unweighted altruism network design, where \EdgeSetCost{k} represents edge cost.}
\label{tab:complexity}
\end{table*}
}

While there are contexts in which a principal can precisely control the amount of effort invested in different actions, there are many others in which actions take more of an all-or-nothing nature.
Indeed, many activities that increase altruism through making harms to others more salient, such as community meetings or public health advertising campaigns (e.g., to promote mask-wearing), are naturally discrete (a 1-second ad is not very effective) and, thus, have a discrete impact on the altruism graph.
This motivates the special case in which all entries of \MODV must be binary, corresponding to decisions whether or not the principal will add/remove the edge sets \EdgeSet{k}.
As we show presently that even if we restrict \EdgeSet{k} to affect a single edge the problem now becomes NP-hard, we devote the sequel to this special case, and seek to identify what additional structure is sufficient to make the problem, or its approximation, tractable.

%When \EdgeSet{k} is a single edge, the input is a (directed or undirected) altruism graph \GraphIn, and the principal adds/removes edges, resulting in a (directed or undirected) graph $G$.
%  Recall that in that case, the altruism from $i$ to $j \neq i$ (with $(i,j) \in E$) is a fixed constant \ALT.

%For example, it would be impractical to schedule a community meeting
%for 30 seconds or several minutes.
%\sydelete{
  %--- discussed in Section~\ref{sec:model} ---
%  As discussed in Section~\ref{sec:model} as well, in this case, we restrict our attention to the special case in which each set \EdgeSet{k} is a single edge.
%  Thus, the input is a (directed or undirected) altruism graph \GraphIn, and the principal adds/removes edges, resulting in a (directed or undirected) graph $G$.
%  Recall that in that case, the altruism from $i$ to $j \neq i$ (with $(i,j) \in E$) is a fixed constant \ALT.
%}
%\syedit{
%  This motivates the special case in which all entries of \MODV must be binary, and correspond to decisions whether or not the principal will add/remove the edge sets \EdgeSet{k}.
%  Due to the graph-theoretic nature of the problem, we will then also always restrict our attention to the case where the principal adds/removes individual edges only, i.e., each set \EdgeSet{k} is a singleton $(i,j)$.
%In this case,
When \EdgeSet{k} is a singleton $(i,j)$, 
we can think of the input as a graph $\GraphIn = (V, \EdgeIn)$ (instead of a weighted network). There is a cost $\EdgeSetCost{i,j}$ associated with each (directed or undirected) node pair $(i,j)$. If $(i,j) \in \EdgeIn$, then this is the cost for removing $(i,j)$; otherwise, it is the cost of adding $(i,j)$. Thus, implicitly, $\EdgeSetSign{i,j} = 1$ if $(i,j) \notin \EdgeIn$, and $\EdgeSetSign{i,j} = -1$ if $(i,j) \in \EdgeIn$.
  Adding/removing edges results in a new altruism network $\GraphOut = (V, \EdgeOut)$. 
  All off-diagonal non-zero entries of the altruism network \ALTM then have the same altruism value \ALT, while all diagonal entries are set to 1.
%}

We study two variants of this problem: 1) asymmetric altruism, that
is, when the altruism graph is directed, and 2) symmetric altruism,
where it is undirected.
Typically, both will capture some important aspect of the real world: while altruism often aligns with actual social or kinship ties, it can also result from a general sense of responsibility or goodwill, which may not be reciprocated.

%While considering only the addition/removal of single edges may appear very restrictive, we first observe that even in this restrictive special case, the integral \ANM is intractable in general.
%We proceed by systematically analyzing the algorithmic complexity of several natural restrictions of the \ANM problem in which only single edges are added/removed.
% Furthermore, we exhibit an FPTAS for the general variant of asymmetric altruism.

\subsection{Asymmetric Altruism}

We begin by formally showing that in this setting, \ANM is in general
intractable, even in the special case where we can only add or remove
individual edges, rather than subsets of edges.

\begin{theorem}\label{th:hard-directed}
  \ANM with asymmetric altruism is NP-complete even when:
  \begin{enumerate}
  \item the sets \EdgeSet{k} are singletons and can only be added, i.e., $\EdgeSetSign{k} = 1$ for all $k$,
  \item the initial altruism network is empty, and
  \item the target profile \EISV has all agents investing.
  \end{enumerate}
\end{theorem}
The proof is a direct reduction from the \textsc{Knapsack} problem and given in the Supplementary Material.

%Given this result, we henceforth restrict attention to the special
%case of our problem in which \EdgeSet{i} are singletons, that is,
%where we can only add or remove edges.
%In this version, we now present a series of positive results.
%The key insight is that not only can our problem be reduced from
%\textsc{Knapsack} in this setting, but we can also naturally model it
%using the \textsc{Knapsack} problem.
Next, we show that under mild additional assumptions, the problem becomes efficiently solvable.
The key observation that enables our positive results is that for directed graphs, adding or removing an edge $\DiEdge{i}{j}$ only affects agent $i$'s altruistic behavior.
This allows us to decompose the problem into $n$ independent subproblems only connected through a common budget constraint.
Each subproblem can be naturally modeled as a \textsc{Knapsack} problem, and so long as the \textsc{Knapsack} problems are individually solvable, so is the overall problem.
More specifically, we distinguish two cases, based on whether $i \in \InvPlayer$.
\begin{itemize}
\item If $i \in \InvPlayer$, then the altruism edges originating with $i$ must ensure that
$\sum_{j \in \Neigh[H]{i}} \Alt{i}{j} \Diff{j}{-} \geq \Threshold{i}$.
Let $\phi_i = \sum_{j \in \Neigh[H]{i}} \AltIn{i}{j} \Diff{j}{-}$.
If $\phi_i \geq \Threshold{i}$, then $i$ will invest even without adding any edges.
Otherwise, the principal will need to add edges out of $i$ adding a total altruism term of $\Threshold{i} - \phi_i$; also, the principal will never want to remove edges out of $i$.
Adding a directed edge \DiEdge{i}{j} can be thought of as putting an item with value $\ALT \cdot \Diff{j}{-}$ and cost/weight \EdgeSetCost{i, j} into a knapsack.
Thus, the set of ``items'' available to add to agent $i$ is $\Pack{i} = \Set{j \in \Neigh[H]{i}}{\AltIn{i}{j} = 0}$.
The subproblem is then to select items from \Pack{i} such that the total value is at least $\Threshold{i} - \phi_i$ while the total weight is minimized. 
\smallskip
\item Similarly, if $i \notin \InvPlayer$, then the altruism edges originating with $i$ must ensure that
$\sum_{j \in \Neigh[H]{i}} \Alt{i}{j} \Diff{j}{+} \leq \Threshold{i}$.
Let $\phi_i = \sum_{j \in \Neigh[H]{i}} \AltIn{i}{j} \Diff{j}{+}$.
Analogously to the previous case, the principal now wants to \emph{remove} edges such that the altruism is reduced by at least $\phi_i - \Threshold{i}$ (unless $\phi_i \leq \Threshold{i}$, in which case nothing needs to be done).
Again, this problem can be modeled as a \textsc{Knapsack} problem.
The set of items available is $\Pack{i}=\Set{j \in \Neigh[H]{i}}{\AltIn{i}{j} = 1}$.
The directed edge \DiEdge{i}{j} is modeled as an item with value $\ALT \cdot \Diff{j}{+}$ and cost/weight \EdgeSetCost{i, j}.
The goal is to minimize the total weight subject to achieving total value at least $\phi_i - \Threshold{i}$.
\end{itemize}

It is well known \citep{kleinberg:tardos:algorithm-design,vazirani:approximation-algorithms} that the \textsc{Knapsack} problem can be solved in polynomial time using Dynamic Programming when either the weights or the values are bounded by a polynomial in the number of items. Using standard rounding/scaling techniques, this approach also yields an FPTAS.
%\footnote{Fully Polynomial-Time Approximation Scheme: an algorithm parametrized by some $\epsilon > 0$ that runs in time $\text{poly}(1/\epsilon, n)$ and produces a $(1+\epsilon)$ approximation.}
By leveraging these algorithms, we obtain the corresponding results for our problem.

\noindent{\bf General $\GFUNC$ and Polynomial Edge Costs }
When edge costs are bounded by a polynomial in $n$, which corresponds to polynomially bounded item weights in the \textsc{Knapsack} instances, which are therefore polynomial-time solvable with Dynamic Programming (DP).
Applying DP for each agent $i$ separately then yields a minimum-cost overall solution. We obtain the following:

%For an agent $i \in \InvPlayer$, we keep adding directed edges \Edge{i}{j} until the altruistic behavior is above \Threshold{i}, starting from the item with the largest $\Diff{j}{-}$.
%Similarly, for an agent $i \notin \InvPlayer$ we keep deleting edges \Edge{i}{j} in descending order of $\Diff{j}{+}$ until the altruistic behavior is below \Threshold{i}. 
%The formal statement is as follows:
\begin{proposition}
   Under asymmetric altruism, the problem \ANM with general $\GFUNC$ and polynomially bounded edge costs is polynomial-time solvable.
\end{proposition}
    
%The case with uniform $\Delta \GFUNC$ and uniform edge costs can be similarly solved with the greedy algorithm described above. 
%Next, we consider another special case with polynomial-bounded $\Delta \GFUNC$ and general edge costs.

\noindent{\bf Polynomial $\GFUNC$ and General Edge Costs }
When all the \GFUNC are polynomially bounded, so are their differences, and hence the (scaled) item values \Diff{j}{+} and \Diff{j}{-}.
Hence all the values of the ``items'' are polynomially bounded.
Tractability then follows from the fact that the \textsc{Knapsack} problem is polynomial-time solvable (using Dynamic Programming) when item values are bounded by some polynomial in the length of the input.
%In particular, for an agent $i \in \InvPlayer$ (resp. $i \notin \InvPlayer$), the subproblem is a \textsc{Knapsack} problem with item values $\Diff{i}{-}$ (resp. $\Diff{i}{+}$) and weights \EdgeCost{\DiEdge{i}{j}}.
Again, this allows an algorithm to solve each subproblem in polynomial time, and then aggregate the optimal solutions. 
%The overall time complexity is clearly polynomial.
%We omit the details for solving the \textsc{Knapsack} problem with DP.
%Please refer to~\citep{kleinberg2006algorithm}.
This gives rise to the following proposition:
\smallskip
\begin{proposition} \label{prop:directed-bounded-poly}
Under asymmetric altruism, the problem \ANM with polynomially bounded $\GFUNC$ and general edge costs is polynomial-time solvable.
\end{proposition}

\noindent{\bf An FPTAS for the general case }
Finally, we can leverage the standard FPTAS for \textsc{Knapsack}
to obtain an FPTAS for \ANM with general asymmetric altruism.
Specifically, given a parameter $\epsilon$, one can run the FPTAS with that parameter for each of the subproblems/agents $i$ separately.
The result for each $i$ will be a set of edges to add/remove such that $i$ invests iff $\EIS{i} = 1$, and the total cost of the modifications is within a factor  $(1+\epsilon)$ of optimal. Adding all of these costs shows that the overall cost is within a factor $(1+\epsilon)$ of optimal.
We obtain the following proposition:

% for an agent $i$ the FPTAS runs a DP on the subproblem with properly scaled item values (either $\Diff{i}{-}$ or $\Diff{i}{+}$), which gives a cost $B_i$ that is arbitrarily close to the optimal cost $B^\ast_i$, i.e., $B_i \le (1+\epsilon_i) B^\ast_i$, where $\epsilon_i > 0$ is an input parameter.
% The approximation algorithm is to run the FPTAS on all subproblems.
% Let $B^\ast=\sum_{i \in V}^{}{B^\ast_i}$ be the optimal total cost.
% The total cost computed by the approximation algorithm is  $B = \sum_{i \in V}^{}{B_i} \le  \sum_{i \in V}^{}{(1+\epsilon_i)B^\ast_i} \le (1+\epsilon) B^\ast$, where $\epsilon=\max_{i} \epsilon_i$.
% Thus, by selecting $\epsilon$ the cost $B$ can be made arbitrarily close to $B^\ast$.
% The formal statement is as follows:

\begin{proposition} \label{prop:approx-directed}
   Under asymmetric altruism, consider \ANM with general $\GFUNC$ and general edge costs.
   Given $\epsilon > 0$, the optimal cost $B^*$ can be approximated arbitrarily well, i.e., a solution of cost $B \leq (1+\epsilon)B^*$ can be found, with an algorithm which runs in time polynomial in $n$ and $1/\epsilon$.
\end{proposition}

    % Next, we consider several tractable special cases of \ANM in the asymmetric
    % altruism setting.
    %other variants of NDA with different combinations of marginal benefit functions and edge costs. 
%These combinations generate various algorithmic implications.

\subsection{Symmetric Altruism}

Next, we turn to the setting in which altruism is reciprocal or symmetric: when an edge is added to the altruism network, it affects \emph{both} incident agents. While many graph-theoretic questions are easier for undirected graphs than for directed ones, the \ANM problem becomes harder. Intuitively, the reason is that while adding the edge $(i,j)$ is very beneficial for $i$, it may be less so for $j$; given the choice, adding a different edge out of $j$ may be preferable.
Under the symmetric altruism model, the principal does not have the fine-grained control of adding different edges out of $i$ and $j$, and might have to ``waste'' one direction of the edge. The resulting ``side-effects'' of desirable edges must be more globally balanced. Indeed, we show that even special cases that are polynomial-time solvable in the asymmetric model become NP-hard in the symmetric model.
\smallskip
\begin{theorem}\label{th:hard-undirected}
  \ANM with symmetric altruism is NP-complete even when
  \begin{itemize}
  \item all the sets \EdgeSet{k} are singletons,
  \item all agents invest under the target equilibrium \EISV,
  \item all $\GFUNC$ are polynomially bounded,
  \item all edge costs \EdgeSetCost{k} are 1, and
  \item the graph $H$ is a clique.
  \end{itemize}
\end{theorem}
The proofs of this and the remaining results are in the Supplementary Material.
Recall that for the asymmetric case, even just one of uniform (or even just polynomially bounded) edge costs and polynomially bounded $\GFUNC$ is enough to obtain tractability.

\dkedit{
  For the remainder of this section, we will focus on the special case when all agents invest under the target equilibrium \EISV. This is because for more general target equilibria, even deciding if there exists \emph{any} altruism graph yielding this equilibrium is NP-complete. In other words, even the decision problem of a principal with infinite budget is NP-complete.
%  Although it is natural to consider target equilibria \EISV other than that all agents invest, we show in the following theorem that for an arbitrary \EISV even if checking the feasibility of \ANM is NP-complete.
%  This result indicates that we should focus on the target equilibria that all agents invest. 

\begin{theorem}\label{th:fea-hard-undirected}
  \ANM with symmetric altruism for arbitrary target equilibria \EISV is NP-complete even when
  \begin{itemize}
  \item all the sets \EdgeSet{k} are singletons,
  \item all $\GFUNC$ are polynomially bounded,
  \item all edge costs \EdgeSetCost{k} are 0 (i.e., the principal's budget is infinite),
  \item the graph $H$ is a clique.
  \end{itemize}
\end{theorem}
}

\dkedit{Theorem~\ref{th:fea-hard-undirected} implies that no approximation guarantee can be attained in polynomial time for \ANM when \EISV is an arbitrary action profile.
However, we remark that when all agents invest under \EISV, there is a straightforward polynomial-time $(2+\epsilon)$-approximation algorithm: the algorithm applies the FPTAS to the asymmetric version and adds a reciprocal edge whenever a directed edge is added. This leads to a blowup of a factor 2 compared to the FPTAS achieved in the asymmetric setting.
}

\smallskip
\noindent{\bf Uniform Separable Linear Utility Functions }
When the agents have utility functions \GFUNC that are separable in the arguments, and linear with common slope in the second argument (we term these uniform separable and linear, or \textsc{USL}), the marginal benefits are uniform and equal to a constant \MarginB, i.e., $\Diff{i}{-}=\Diff{i}{+}=\MarginB$ for all $i$ and all values of $n_i$.
Such utility functions are commonly studied in public goods games~\citep{suri2011cooperation}.
We show that when the goal is for all players to invest, for \textsc{USL} utility functions and when all sets \EdgeSet{k} are single edges, the problem becomes tractable, even with all other parameters (edge costs, network structure, etc.) being fully general.

    \begin{theorem}\label{th:tract}
      \ANM with symmetric altruism is polynomial time solvable when
        \begin{itemize}
        \item the utility functions \GFUNC are USL,
        \item the sets \EdgeSet{k} are singleton,
        \item the target equilibrium \EISV has all agents investing.
        \end{itemize}
    \end{theorem}
  The proof of Theorem~\ref{th:tract} (given in the Supplementary Material) is through a non-trivial reduction to the \textsc{Min-Cost Perfect Matching} problem via a connection to another related problem: \textsc{Network Design for Degree Sets} (NDDS)~\citep{kempe2020inducing}.

\section{Conclusion}\label{sec:conclusion}
We consider how to change altruistic behavior of individuals so as to induce societally desirable outcomes.
One major contribution of our work is to separately capture the strategic interdependencies and the altruistic behaviors of individuals. 
We propose a model of modifying the altruism network, with the goal of inducing a target investment profile by the individuals.
A series of corresponding algorithmic results are exhibited, including hardness results even in very restrictive scenarios (e.g., each modification only affects a single edge), and tractability results in a broad array of special cases.

%One future research direction is to generalize the model of altruistic behaviors to other game-theoretic models.
%Another future direction is motivated by the following observation in the case of symmetric altruism:
%if the principal wishes for all players to invest (i.e., \InvPlayer is the set of all players),
%there is a straightforward $2(1+\epsilon)$-factor approximation algorithm for \ANM that comes from applying the FPTAS in the asymmetric setting and adding a reciprocal edge whenever a directed edge is added.
%However, this algorithm does not work for a general target \EISV,
%since adding reciprocal edges may cause some of the non-investing players to invest.
%It is an open question whether we can obtain a constant factor approximation algorithm for \ANM in the symmetric altruism setting with a general \EISV.

Our work only focused on the goal of achieving a \emph{specific} strategy profile \EISV. 
Other natural goals would be to maximize the social welfare of the resulting profile (e.g., subject to a budget constraint), or to maximize the number of investing players. In these cases, many of our current reductions will fail, since they rely on knowing exactly which set of players will invest.
More fundamentally, the simultaneous study of externality and altruism networks raises interesting structural questions. 
Is there a sense in which for some broad class of games (such as public goods games), it is desirable for these networks to be aligned? In other words, should individuals care most about those who are most affected by their actions? Or can it be beneficial to have ``long-range'' edges?
Additionally, our model of altruism assumes that edges can be added/removed by the principal, so long as the principal is willing to pay for the action. A natural alternative would be to recognize that the total ``capacity'' of an agent for altruism towards others may be bounded. In other words, the principal may be able to introduce new altruism edges, but the more outgoing edges agent $i$ has, the weaker each of them becomes. 
This change will require different algorithmic and structural insights.

% \dkcomment{It would be nice if we can save enough space to keep the second paragraph --- people often like open questions to work on themselves, and might be more inclined to accept the paper so it's public and they can start working on the questions.}

\begin{small}
\balance
\bibliographystyle{named}
%% The first four files of the following are borrowed from David's database
\bibliography{names,conferences,bibliography,publications,main}
\end{small}

\newpage
\appendix
\section*{Supplementary Material}
\section{Proof of Theorem~\ref{th:hard-directed}}
%\begin{proof}
   We consider the decision version of \ANM with an input parameter $B$, the available budget.
   The goal is to decide if the total cost of producing a graph \GraphOut with $\EISV \in \ALLPSNE{\BNPG{\GraphOut}}$ exceeds $B$.
   It is clear that the problem is in NP: given a candidate set of actions, their costs can be added, and one can verify that the desired equilibrium \EISV is indeed a PSNE.

   To prove NP-harness, we reduce from the \textsc{Knapsack} problem.
   Recall that in the \textsc{Knapsack} problem, we are given a knapsack with capacity $W$, a set of $n$ items, and a value $V$.
   Each item has a given weight $w_i \ge 0$ and a value $v_i \ge 0$. 
   The goal is to decide if there is a subset $S$ of items, such that the total weight $\sum_{i \in S} w_i \leq W$ and
   the total value $\sum_{i \in S} v_i \geq C$.

   Given an instance of the \textsc{Knapsack} problem, we construct an instance of \ANM as follows.
   The target investment profile \EISV is that all players invest.
   The altruism constant $\ALT > 0$ is arbitrary (e.g., $\ALT = 1$).
   The graph $H=(V, E_H)$ is a clique with $n+1$ nodes; a special node $u_c$ together with $n$ nodes $u_1, \ldots, u_n$.
   The initial altruism network \GraphIn is empty, i.e., $\AltIn{u_i}{u_j} = 0$ for all $i \neq j$.
   For the node $u_c$, let $\Threshold{u_c}=\ALT \cdot C$ and $\Diff{u_c}{-}$ be an arbitrary nonnegative number.
   For each node $u_i$, let $\Threshold{u_i}=0$ and $\Diff{u_i}{-} = v_i$.
   Intuitively, the node $u_i$ corresponds to the item with value $v_i$.

   To realize the above, the function \GFUNC[u_c] is constant over the whole domain, i.e., $\Delta_{x_{u_c}} \GFUNC[u_c] = 0$ and $\Diff{u_c}{-}=0$. The cost $\InvCost{u_c}= \ALT \cdot C$. This ensures that $\Threshold{u_c}=\ALT \cdot C$.
   For a node $u_i$, we let \GFUNC[u_i] be a linear function over the first argument with slope \InvCost{u_i}, such that $\Delta_{x_i}\GFUNC[x_i]=\InvCost{u_i}$.
   The function \GFUNC[u_i] over the second argument is also a linear function with slope $v_i$, i.e., $\Diff{u_i}{-}=v_i$.

   The cost to add an edge from $u_c$ to $u_i$ is $w_i$.
   Neither edges from $u_i$ to $u_c$ nor edges between different $u_i$ can be added, i.e., their costs are infinite.
   The budget is the weight capacity, $B = W$.
   The reduction clearly can be constructed in polynomial time.
   
   First, suppose that there is a set $S$ of items that solves the \textsc{Knapsack} problem.
   Then, picking the actions/edges $(u_c, u_i)$ for $i \in S$ does not exceed the budget $B$.
   Each node $u_i$ will invest as $\Threshold{u_i}=0$ and $\sum_{j \in \Neigh[H]{u_i}} \Alt{i}{j} \Diff{j}{-} = 0 \geq \Threshold{u_i}$.
   The node $u_c$ will not deviate from investing since
   $\sum_{u_i \in \Neigh[H]{u_c}} \Alt{u_c}{u_i} \Diff{u_i}{-} = \ALT \cdot \sum_{i \in S} v_i \geq \ALT \cdot C = \Threshold{u_c}$.

   For the converse direction, suppose that there is a subset $S$ of indices $i$ such that the actions/edges $(u_c, u_i)$ for $i \in S$ have total cost at most $B=W$, and such that adding the edges $(u_c, u_i)$ makes \EISV a PSNE. 
   Because $u_c$ has no incentive to deviate, $\sum_{i \in S} v_i \geq C$.
   Thus, the set $S$ solves the \textsc{Knapsack} problem.

\section{Proof of Theorem~\ref{th:hard-undirected}}
% \begin{proof}
  We consider the decision version of \ANM with an input parameter $B$.
  It is clear that the problem is in NP. 
  To prove hardness, we reduce from the \textsc{3-Partition} problem.
  In the \textsc{3-Partition} problem, we are given a set $S=\SET{x_1, \ldots, x_{3m}}$ of $3m$ numbers $x_i$, each bounded by a polynomial in $m$.
  Writing $s = \sum_{i} x_i$ for the sum of all numbers, the numbers further satisfy that $\frac{s}{4m} < x_i < \frac{s}{2m}$.
  The goal is to decide if $S$ can be partitioned into disjoint triples $T_1, \ldots, T_m$ such that each triple
  \footnote{The restriction on the $x_i$ ensures that if there is \emph{any} way to partition the numbers into $m$ sets of equal sum, each set must be a triple.} has the same sum, i.e., $\sum_{x_k \in T_j} x_k = \frac{1}{m} \cdot \sum_{k=1}^{3m} x_k$.
  The \textsc{3-Partition} problem is strongly NP-complete~\citep{garey:johnson}, meaning that it is NP-complete even when the $x_i$ are polynomially bounded in $m$.\footnote{Recall that the \textsc{Knapsack} problem we reduced from earlier is NP-hard only when the values and weights can be exponentially large in the number of items.}
    
  Given an instance of the \textsc{3-Partition} problem, we construct an instance of the \ANM problem as follows.
  Let $\epsilon$ be a small (though polynomial) constant, and $\ALT > 0$ an arbitrary constant.
  The set $V$ of agents comprises two types: a set $V_1$ of $3m$ agents $u_i$ (each corresponding to one number $x_i$) and a set $V_2$ of $m$ agents $v_j$ (each corresponding to one triple $T_j$).
  The target investment profile \EISV is that all agents invest.
  The graph $H$ is a clique on $V$.
  The input altruism network \GraphIn is a complete graph on $V_1$.
  
  For each node $u_i \in V_1$, let $\Diff{u_i}{-} = x_i$, and $\Threshold{u_i} = \ALT \cdot (s + \epsilon - x_i)$.
  For each node $v_j \in V_2$, let $\Diff{v_j}{-} = \epsilon$ and $\Threshold{v_j} = \ALT \cdot s/m$. 
  This setup can be realized by, e.g., setting $\GFUNC[u_i]$ and $\GFUNC[v_j]$ to be linear functions with slopes $x_i$ and $\epsilon$, respectively, w.r.t. the second argument.
  In addition, the $\GFUNC[u_i]$ and $\GFUNC[v_j]$ are constant functions w.r.t. the first argument and $\InvCost{u_i}=a\cdot(s+\epsilon-x_i)$ and $\InvCost{v_j}=a\cdot s/m$, which ensures that we have the above $\Threshold{u_i}$ and $\Threshold{v_j}$.
  The addition/deletion cost for each edge is $1$, and the total budget for additions/deletions is $3m$.  
  The reduction obviously takes polynomial time.

  First, suppose that there are triples $T_1, \ldots, T_m$ such that $\sum_{x_k \in T_j} x_k = s/m$ for all $j$. Consider adding the edges $(u_i, v_j)$ if and only if $x_i \in T_j$. This adds exactly $3m$ edges, so it satisfies the budget constraint. Each node $u_i \in V_1$ is incident on exactly one new edge to a node $v_j$ with $\Diff{v_j}{-} = \epsilon$.
    $u_i$ will not deviate from investing as in the resulting altruism network \GraphOut (where recall that all off-diagonal entries equal the constant \ALT),

    \begin{align*}
      \sum_{j \in V \setminus \SET{u_i}} \Alt{u_i}{j} \Diff{j}{-}
      & = \ALT (\epsilon + \sum_{j \neq i} x_j)
      \; = \; \Threshold{u_i}.
    \end{align*}
    No node $v_j$ will deviate from investing since
    \begin{align*}
      \sum_{i \in V \setminus \SET{v_j}} \Alt{j}{i} \Diff{i}{-}
      & = \sum_{x_i \in T_j} \ALT \cdot \Diff{i}{-}
      \; = \; \ALT \cdot s/m
      \; = \; \Threshold{v_j}.
   \end{align*}

   Next, we show the converse direction.
    Suppose that there is a set of at most $3m$ edges \EdgeOut, such that the modified altruism network is produced by adding \EdgeOut to \ALTMIN and that all agents investing is a PSNE of the resulting game.
    As $\Threshold{u_i} > \ALT \cdot (s-x_i)$ for all $i$, each agent $u_i$ must be incident on at least one edge to some $v_j$. Because there are $3m$ edges in total, each node $u_i$ must be incident on \emph{exactly} one edge to some $v_j$. Thus, we can define a partition $S_1, \ldots, S_m$ of $\SET{x_1, \ldots, x_{3m}}$ via $S_j = \Set{x_i}{u_i \text{ and } v_j \text{ are connected}}$.
    For each $j$, because $v_j$ invests in \EISV (which is a PSNE), we must have that $\ALT \cdot \sum_{x_i \in S_j} x_i \geq \ALT \cdot s/m$.
    But because $\sum_j \sum_{x_i \in S_j} x_i = s$ (because the $S_j$ form a partition), all of the preceding inequalities hold with equality.
    Thus, the $S_j$ form a partition into sets of the same sum, and by the restriction on the $x_i$ values, they must be a partition into triples.
  % \end{proof}
%\end{proof}

\dkedit{
\section{Proof of Theorem~\ref{th:fea-hard-undirected}}
  Given an arbitrary action profile \EISV, checking the feasibility of \ANM is equivalent to deciding if the principal can modify the original altruism network \GraphIn (by adding/removing edges), such that \EISV is an equilibrium of \BNPG{\GraphOut}.
  % \dkcomment{Maybe delete the following sentence --- membership in NP was proved more generally before, I think.}
  % It is clear that checking the feasibility of \ANM is in NP: given a proposed edge set, one only needs to verify that each agent acts as prescribed by \EISV and the total cost does not exceed the budget.
    
  We now show NP-hardness even when the principal has an infinite budget, or --- equivalently --- when all edge costs \EdgeSetCost{k} are 0.
  We prove NP-hardness again by reducing from the 
  \textsc{3-Partition} problem, and use the same notation as in the proof of Theorem~\ref{th:hard-undirected}.
  
  In the \textsc{3-Partition} problem, we are given a set $S=\SET{x_1, \ldots, x_{3m}}$ of $3m$ numbers $x_i$, each bounded by a polynomial in $m$.
  Writing $s = \sum_{i} x_i$ for the sum of all numbers, the numbers further satisfy that $\frac{s}{4m} \leq x_i \leq \frac{s}{2m}$.
  The goal is to decide if $S$ can be partitioned into disjoint triples $T_1, \ldots, T_m$ such that each triple\footnote{The restriction on the $x_i$ ensures that if there is \emph{any} way to partition the numbers into $m$ sets of equal sum, each set must be a triple.} has the same sum, i.e., $\sum_{x_k \in T_j} x_k = \frac{1}{m} \cdot \sum_{k=1}^{3m} x_k$.
    
  Given an instance of the \textsc{3-Partition} problem, we construct an instance of the \ANM problem as follows.
  The set $V$ of agents again comprises two types: a set $V_1$ of $3m$ agents $u_i$ (each corresponding to one number $x_i$) and a set $V_2$ of $m$ agents $v_j$ (each corresponding to one desired triple $T_j$).
  The target investment profile \EISV is that the $m$ agents in $V_2$ invest while the $3m$ agents in $V_1$ do not invest.
  The graph $H$ is again a complete graph on $V_1 \cup V_2$.
  The input altruism network \GraphIn is an empty graph, and, as stated above, all edge additions cost 0.

  For each node $u_i \in V_1$, let $\Diff{u_i}{-} = \Diff{u_i}{+} = x_i$, and $\Threshold{u_i} =  \ALT$.
  For each node $v_j \in V_2$, let $\Diff{v_j}{-} = \Diff{v_j}{+} = 1$ and $\Threshold{v_j} = \ALT \cdot (s/m + m-1)$. 
  This setup can be realized by, e.g., setting $\GFUNC[u_i]$ and $\GFUNC[v_j]$ to be linear functions with slopes $x_i$ and $1$, respectively, w.r.t.~the second argument.
  In addition, the $\GFUNC[u_i]$ and $\GFUNC[v_j]$ are constant functions w.r.t.~the first argument and $\InvCost{u_i}=\ALT$ and $\InvCost{v_j}= \ALT \cdot (s/m+m-1)$, which ensures that we have the above $\Threshold{u_i}$ and $\Threshold{v_j}$.
  As mentioned previously, the addition/deletion cost for each edge is $0$; equivalently, the principal's budget is infinite.
  The reduction obviously takes polynomial time.

  First, suppose that there are triples $T_1, \ldots, T_m$ such that $\sum_{x_k \in T_j} x_k = s/m$ for all $j$. Consider adding the edges $(u_i, v_j)$ if and only if $x_i \in T_j$, and also adding the edges $(v_j, v_{j'})$ for all $j \neq j'$. Each node $u_i \in V_1$is incident on exactly one new edge to a node $v_j$ with $\Diff{v_j}{-} = 1$.
  Not investing is still a best response for $u_i$, because
    \vspace{-0.1in}
    \begin{align*}
      \sum_{j \in V \setminus \SET{u_i}} \Alt{u_i}{j} \Diff{j}{+}
          & = \ALT \cdot 1
          \; = \; \Threshold{u_i};
    \end{align*}
  where recall that all off-diagonal entries equal the constant \ALT.
  Also, no node $v_j$ will deviate from investing since
  \begin{align*}
    \sum_{i \in V \setminus \SET{v_j}} \Alt{j}{i} \Diff{i}{-} & = \sum_{x_i \in T_j} \ALT \cdot \Diff{i}{-} + \ALT \cdot (m-1) \\
                                                              & = \ALT \cdot (s/m + m-1) \\
                                                              & = \Threshold{v_j}.
    \end{align*}

   Next, we show the converse direction.
   Suppose that there is a set of edges \EdgeOut, such that under the altruism network \ALTMIN with edge set \EdgeOut, the profile \EISV is a PSNE of the game.
   Because $\Diff{u_i}{+}, \Diff{v_j}{+} \geq 1$ for all nodes $u_i$ and $v_j$, and $\Threshold{u_i} = \ALT$, each node $u_i$ can be incident on at most one edge in \EdgeOut; otherwise, $u_i$ would deviate to investing.

   Now consider a node $v_j$, which must be investing under \EISV.
   Therefore, 
\begin{align*}
  \sum_{i \in V \setminus \SET{v_j}} \Alt{j}{i} \Diff{i}{-}
  &   \; \geq \; \Threshold{v_j}.
\end{align*}
  By uniform edge weights and the definitions of the \Diff{i}{-}, the sum is equal to $\ALT \cdot (\sum_{u_i: (u_i, v_j) \in \EdgeOut} x_i + \sum_{v_{j'}: (v_{j'}, v_j) \in \EdgeOut} 1)$, and by definition, 
$ \Threshold{v_j} = \ALT \cdot (s/m + m-1)$.
Therefore, we obtain that
$\sum_{u_i: (u_i, v_j) \in \EdgeOut} x_i + \sum_{v_{j'}: (v_{j'}, v_j) \in \EdgeOut} 1 \geq s/m + m-1$.
Because the number of neighbors of $v_j$ in $V_2$ is at most $m-1$, this implies in particular that
$\sum_{u_i: (u_i, v_j) \in \EdgeOut} x_i \geq s/m$.
Summing this inequality over all nodes $v_j \in V_2$, we get that
$\sum_{j} \sum_{u_i: (u_i, v_j) \in \EdgeOut} x_i \geq s$.
But because each node $u_i$ occurs in the sum for at most one node $v_j$, each $x_i$ occurs at most once, so we also get that
$\sum_j \sum_{u_i: (u_i, v_j) \in \EdgeOut} x_i \leq s$, which implies equality, and thus also that
$\sum_{u_i: (u_i, v_j) \in \EdgeOut} x_i = s/m$ for all $j$.
(If the inequality were strict for some $j$, then it would have to be violated for some other $j'$ for the sum to equal $s$.)

Thus, the sets $S_1, \ldots, S_m$ defined via $S_j = \Set{x_i}{u_i \text{ and } v_j \text{ are connected}}$ form a partition of $\SET{x_1, \ldots, x_{3m}}$, such that $\sum_{x_i \in S_j} x_i = s/m$. 
Furthermore, by the restriction on the $x_i$ values, they must be a partition into triples.
}

\section{Proof of Theorem~\ref{th:tract}}
% \dkcomment{Somehow, this proof feels a bit more complicated than it needs to be. I have to think through it more carefully, but it feels like we're just intersecting edge sets (or edge addition sets) with $\mathcal{H}$ and then showing that they're identical. Somehow, that shouldn't take so many steps.}

     We show the theorem through a non-trivial reduction to the \textsc{Min-Cost Perfect Matching} problem via a connection to another related problem: \textsc{Network Design for Degree Sets} (NDDS)~\citep{kempe2020inducing}.
      We begin by introducing the NDDS problem.
      \smallskip
      \begin{definition}[NDDS]\label{def:NDDS}
        Given a simple, undirected, and loop-free graph $\NDSSGraphIn=(V, \NDSSEdgeIn)$,
        costs $\EdgeSetCost{i,j} \ge 0$ for each node pair\footnote{%
      If $(i,j) \in \NDSSEdgeIn$, then this is the cost of removing $(i,j)$; otherwise, it is the cost of adding $(i,j)$.} $(i,j)$, and target degree sets \InvDeg{i} which are intervals, i.e., of the form $\SET{\ell_i, \ell_i + 1, \ldots, r_i}$. 
        The problem NDDS is to change \NDSSEdgeIn to a graph $\NDSSEdgeOut=(V, \NDSSEdgeOut)$ by adding/deleting edges, such that the total cost is minimized and for all nodes $i$, the resulting degree $\Degree[\NDSSGraphOut]{i} \in \InvDeg{i}$.
      \end{definition}

      \citet{kempe2020inducing} showed that the NDDS problem as defined in Definition~\ref{def:NDDS} can be solved in polynomial time.
      The formal statement is below:
      \begin{lemma}[Theorem 4.1 of~\citet{kempe2020inducing}] \label{lemma:NDDS-poly}
         The NDDS problem can be solved in polynomial time, via a reduction to weighted non-bipartite minimum-cost matching.
      \end{lemma}

      When the agents have \textsc{USL} utility functions, the marginal benefits are all identical, i.e., $\Diff{i}{-}=\Diff{i}{+}=\MarginB$.
      In the following discussion, we assume that \MarginB is strictly positive.
      Otherwise, there is no motivation for network changes: whether \EISV is a PSNE of \BNPG{\GraphOut} is solely determined by $\Threshold{1}, \ldots, \Threshold{n}$. 

      The intuition for the reduction is as follows: for any pair of nodes $(i,j)$ which do not share an edge in $H$, an edge in the altruism network is useless, so we make the cost of such edges $(i,j)$ infinite.
      When $i$ and $j$ do affect each other, an edge between them contributes 
      \MarginB to the total altruistic utility experienced by $i$.
      Thus, whether $i$ crosses her investment threshold or not only depends on the \emph{number} of neighbors she has in the altruism network $G$, i.e., on her degree.
      When $i$ is to invest, this number must be large enough (at least a given threshold), while if $i$ is not to invest, it must be small enough (below a given threshold).

      We now describe the reduction more formally.
      Slightly abusing notation, a BNPG with underlying altruism network $G$ is denoted by \BNPG{G}.
      The input is an instance of \ANM, consisting of an input altruism graph $\GraphIn = (V, \EdgeIn)$, the desired PSNE \EISV, and costs \EdgeSetCost{i, j} for edge addition/deletion.
      % , the principal wants to modify \GraphIn to $\Graph=(V, E)$, such that  is a PSNE of \BNPG{\GraphOut} and the total cost is minimized. 
      In what follows, we show how to construct an instance of \NDDS on a graph $\NDSSGraphIn = (V, \NDSSEdgeIn)$ and show that the costs of solving the two instances are equal.
      As suggested before, \NDSSEdgeIn contains the (undirected) edge $(i, j)$ if and only if $j \in \Neigh[H]{i} \cap \Neigh[\GraphIn]{i}$; in other words, edges such that $i$'s actions do not affect $j$ are removed from the altruism network.
      Next, we define the target degree sets $D_i$ for each node.
      When $i \in \InvPlayer$,
      we define $D_i := \SET{\Ceiling{\Threshold{i} / (\ALT \cdot \MarginB )}, \ldots, n-1}$;
      when $i \notin \InvPlayer$,
      the set is $D_i := \SET{0, \ldots, \Floor{\Threshold{i} / (\ALT \cdot \MarginB )}}$.
      Finally, we define the cost \EdgeSetCostTmp{i,j} associated with a node pair $(i, j)$.
      When $(i, j) \in \NDSSEdgeIn$, we set $\EdgeSetCostTmp{i, j} = \EdgeSetCost{i, j}$.
      When $(i, j) \notin \NDSSEdgeIn$, if $i \in \Neigh[H]{j}$, the cost is $\EdgeSetCostTmp{i, j}=\EdgeSetCost{i, j}$; otherwise, it is $\EdgeSetCostTmp{i, j}=\infty$, meaning that the edge $(i,j)$ can not be added.
      It is clear that the construction runs in polynomial time.
      The correctness of the construction is captured by Theorem~\ref{th:construction},
      By Lemma~\ref{lemma:NDDS-poly}, \NDDS is polynomial-time solvable; therefore, we obtain a polynomial-time algorithm for this variant of \ANM.

      \begin{theorem}\label{th:construction}
         The minimum cost of modifying $\GraphIn=(V, \EdgeIn)$ to a graph $\GraphOut=(V, \EdgeOut)$ such that \EISV is a PSNE of \BNPG{\GraphOut} is equal to the minimum cost of modifying $\NDSSGraphIn=(V, \NDSSEdgeIn)$ to a graph $\NDSSGraphOut=(V, \NDSSEdgeOut)$ such that $\Degree[\NDSSGraphOut]{i} \in D_i$ for all $i$.
      \end{theorem}
      \begin{proof}
        Because $\Diff{i}{-} = \Diff{i}{+} = \MarginB$ for all $i$ and $n_i$, the action profile \EISV is a PSNE of \BNPG{\GraphOut} if and only if:
        \begin{equation}\label{eq:suff-neighbor}
           \begin{aligned}
             \SetCard{\Neigh[H]{i} \cap \Neigh[\GraphOut]{i}} & \geq \Threshold{i} / (\ALT \cdot \MarginB)
             \quad \text{ for all } i \in \InvPlayer \\ 
             \SetCard{\Neigh[H]{i} \cap \Neigh[\GraphOut]{i}} & \leq \Threshold{i} / (\ALT \cdot \MarginB)
             \quad \text{ for all } i \notin \InvPlayer.
           \end{aligned}
         \end{equation}     
                      
         For the first direction, let $\GraphOut=(V, \EdgeOut)$ be an undirected unweighted altruism network such that \EISV is a PSNE of \BNPG{\GraphOut}.
         Let $B = \sum_{e \in \EdgeIn \triangle \EdgeOut} \EdgeCost{e}$ be the cost of producing \GraphOut.
         We construct a graph \NDSSGraphOut such that $\Degree[\NDSSGraphOut]{i} \in D_i$ for all $i$.
                  
         Given the graphs $H$ and \GraphOut, the construction of $\NDSSGraphOut=(V, \NDSSEdgeOut)$ is similar to that of \NDSSGraphIn.
         \NDSSGraphOut contains the edge $(i, j)$ iff $j \in \Neigh[H]{i} \cap \Neigh[\GraphOut]{i}$.
         The neighborhood of a node $i$ is thus $\Neigh[\NDSSGraphOut]{i}=\Set{j}{j \in \Neigh[H]{i} \cap \Neigh[\GraphOut]{i}}$.
         The degree of the node is $\Degree[\NDSSGraphOut]{i} = \SetCard{\Neigh[H]{i} \cap \Neigh[\GraphOut]{i}}$.
         As \EISV is a PSNE of \BNPG{\GraphOut}, Equation~\eqref{eq:suff-neighbor} implies that for a player $i \in \InvPlayer$ (respectively, $i \notin \InvPlayer$), the degree satisfies $\Degree[\NDSSGraphOut]{i} \geq \Ceiling{\Threshold{i}/ (\ALT \cdot \MarginB)}$ (resp. $\Degree[\NDSSGraphOut]{i} \leq \Floor{\Threshold{i}/ (\ALT \cdot \MarginB)}$).
         Thus, the degree $\Degree[\NDSSGraphOut]{i} \in D_i$ for all $i$.
                  
         Now, we show that the cost of modifying \NDSSGraphIn to \NDSSGraphOut is equal to $B$. 
         Let $\EdgeChange := \EdgeOut \triangle \EdgeIn$ be the modification from \GraphIn to \GraphOut.
         Similarly, the modification from \NDSSGraphIn to \NDSSGraphOut is  $\NDSSEdgeChange :=  \NDSSEdgeOut \triangle \NDSSEdgeIn$.
         We now show that $\EdgeChange = \NDSSEdgeChange$, in particular implying that the proposed modification has cost at most $B$.
                  
         First, observe that every edge $(i, j) \in \EdgeChange$ must be in $H$. Otherwise, adding/deleting $(i,j)$ has no impact on altruistic behavior, so $B$ would not have minimum cost. We now reason as follows:
         \begin{itemize}
           \item By definition, $(i, j) \in \EdgeOut \setminus \EdgeIn$ iff $i$ and $j$ are neighbors in $H$ and \GraphOut, but not in \GraphIn.
           This is the case if and only if $(i,j)$ are not neighbors in \NDSSGraphIn, but are neighbors in \NDSSGraphOut, which is the case if and only if $(i, j) \in \NDSSEdgeOut \setminus \NDSSEdgeIn$.
         \item Similarly, $(i, j) \in \EdgeIn \setminus \EdgeOut$ iff $i$ and $j$ are neighbors in $H$ and \GraphIn, but not in \GraphOut.
           This is the case if and only if $(i,j) \in \NDSSEdgeIn \setminus \NDSSEdgeOut$
         \end{itemize}
         Therefore, we have shown that $\EdgeChange = \NDSSEdgeChange$

         For the converse direction, we let \NDSSGraphOut be the minimum-cost modification of \NDSSGraphIn such that $\Degree[\NDSSGraphOut]{i} \in D_i$ for all $i$.
         Let $B$ be the cost of modifying $\NDSSGraphIn$ to \NDSSGraphOut.
         We first observe that every edge $(i, j) \in \NDSSEdgeChange$ must be in $H$.
         If $(i, j) \in \NDSSEdgeIn \setminus \NDSSEdgeOut$, this follows from the definition of \NDSSEdgeIn, whereas if $(i, j) \in \NDSSEdgeOut \setminus \NDSSEdgeIn$, it is because the cost of adding $(i,j)$ would be infinite if $(i,j)$ were not in $H$.

         Next, we construct the graph \GraphOut.
         By the construction of \NDSSGraphIn, $(i, j) \in \NDSSEdgeIn \setminus \NDSSEdgeOut$ implies that $i$ and $j$ are connected in \GraphIn, whereas $(i, j) \in \NDSSEdgeOut \setminus \NDSSEdgeIn$ implies that $i$ and $j$ are \emph{not} connected in \GraphIn.
         In the former case, $(i, j)$ is deleted from \GraphIn, i.e., $(i, j) \in \EdgeIn \setminus \EdgeOut$, while in the latter case, $(i, j)$ is added to \GraphIn, i.e., $(i, j) \in \EdgeOut \setminus \EdgeIn$.
         This modification is applied to all edges in $\NDSSEdgeChange$, producing a graph \GraphOut at cost $B$.
                  
         It remains to show that \EISV is a PSNE of \BNPG{\GraphOut}.
         Consider any node $i \in V$.
         By construction of \GraphOut, the set of neighbors added for $i$ is
         $\Neigh[H]{i} \cap (\Neigh[\GraphOut]{i} \setminus \Neigh[\GraphIn]{i})$.
         Similarly, the set of neighbors removed from $i$ is
         $\Neigh[H]{i} \cap (\Neigh[\GraphIn]{i} \setminus \Neigh[\GraphOut]{i})$.
         % We use $\Delta^+ \Neigh[\NDSSGraphIn]{i}$ (resp. $\Delta^- \Neigh[\NDSSGraphIn]{i}$) to represent the set of nodes that are added (resp. deleted) to $i$'s neighborhood in $\NDSSGraphIn$.

         %              \begin{equation}
         %                  \begin{aligned}
         %                      \Delta^+ \Neigh[\NDSSGraphIn]{i} & = \Set{j}{(i,j) \in E_1 \setminus E'_1} \\
         %                                               & = \Set{j}{(i,j) \in E \setminus E'} \\
         %                                               & = \SET{\Neigh[H]{i} \cap (\Neigh[G]{i} \setminus \Neigh[G']{i})} \\
         %                      \Delta^- \Neigh[\NDSSGraphIn]{i} & = \Set{j}{(i,j) \in E'_1 \setminus E_1} \\
         %                                               & = \Set{j}{(i,j) \in E' \setminus E} \\
         %                                               & = \SET{\Neigh[H]{i} \cap (\Neigh[G']{i} \setminus \Neigh[G]{i})}
         %                  \end{aligned}
         %              \end{equation}
         The degree of node $i$ in the graph \NDSSGraphOut can be expressed in these terms as 
         \begin{align*}
           \Degree[\NDSSGraphOut]{i}
           & = \Degree[\NDSSGraphIn]{i}
             + \SetCard{\Neigh[H]{i} \cap (\Neigh[\GraphOut]{i} \setminus \Neigh[\GraphIn]{i})} \\
             & - \SetCard{\Neigh[H]{i} \cap (\Neigh[\GraphIn]{i} \setminus \Neigh[\GraphOut]{i})} \\
           & = \SetCard{\Neigh[H]{i} \cap \Neigh[\GraphIn]{i}}
             + \SetCard{\Neigh[H]{i} \cap (\Neigh[\GraphOut]{i} \setminus \Neigh[\GraphIn]{i})} \\
             & - \SetCard{\Neigh[H]{i} \cap (\Neigh[\GraphIn]{i} \setminus \Neigh[\GraphOut]{i})} \\
           & = \SetCard{\Neigh[H]{i} \cap \Neigh[\GraphOut]{i}}.
        \end{align*}
        As $\Degree[\NDSSGraphOut]{i} \in D_i$, it follows that $\SetCard{\Neigh[H]{i} \cap \Neigh[\GraphOut]{i}} \geq \Ceiling{\Threshold{i} / (\ALT \cdot \MarginB)}$ for $i \in \InvPlayer$, and 
        $\SetCard{\Neigh[H]{i} \cap \Neigh[\GraphOut]{i}} \leq \Floor{\Threshold{i} / (\ALT \cdot \MarginB)}$ for $i \notin \InvPlayer$.
        Thus, \EISV is a PSNE of \BNPG{\GraphOut}.
      \end{proof}

\end{document}